\documentclass[12pt]{article}
\usepackage{amsmath,amsthm}
\usepackage{amssymb}
\usepackage{amsfonts}
\usepackage{amscd}
\usepackage{color}
\usepackage{mnsymbol}
\usepackage{graphicx}
\usepackage{xfrac}

\newtheorem{thm}{Theorem}[section]

\newtheorem{proposition}[thm]{Proposition}
\theoremstyle{definition}
\newtheorem{definition}[thm]{Definition}

\newtheorem{remark}[thm]{Remark}

\newcommand{\C} {\mathbb{C}}

\newcommand{\R}{\mathbb{R}}

\newcommand{\bh}{\mathbf{h}}

\newcommand{\Us}{\cU_q^{(s)}\hspace{-0.5mm}(\rsl_2)}
\newcommand{\Uj}{\cU_p^{(j)}\hspace{-0.5mm}(\rsl_2)}

\newcommand{\Ss}{\rSL^{\hspace{-0.6mm}(s)}_q\hspace{-0.6mm}(2,\C)}
\newcommand{\Sj}{\rSL^{\hspace{-0.6mm}(j)}_p\hspace{-0.6mm}(2,\C)}

\newcommand{\cH}{\mathcal {H}}

\newcommand{\cU}{\mathcal{U}}

\newcommand{\al}{\alpha}
\newcommand{\be}{\beta}
\newcommand{\ga}{\gamma}
\newcommand{\de}{\delta}
\newcommand{\ep}{\epsilon}

\newcommand{\fg}{\mathfrak g}

\newcommand{\se}{*_{\hspace{-0.4mm}\textrm{\tiny{e}}}}
\newcommand{\sre}{\star_{\hspace{-0.4mm}\textrm{\tiny{e}}}}
\newcommand{\sk}{*_{\hspace{-0.4mm}\textrm{\tiny{k}}}}
\newcommand{\srk}{\star_{\hspace{-0.4mm}\textrm{\tiny{k}}}}
\newcommand{\sts}{*_{\hspace{-0.4mm}\textrm{\tiny{s}}}}
\newcommand{\stk}{*_{\hspace{-0.4mm}\textrm{\tiny{k}}}}
\newcommand{\ste}{*_{\hspace{-0.4mm}\textrm{\tiny{e}}}}
\newcommand{\srs}{\star_{\hspace{-0.4mm}\textrm{\tiny{s}}}}
\newcommand{\stl}{*_{\hspace{-0.4mm}\textrm{\tiny{l}}}}

\newcommand{\srl}{\star_{\hspace{-0.4mm}\textrm{\tiny{l}}}}
\newcommand{\Uqs}{\cU_q^{(s)}\hspace{-0.5mm}(\rsl_2)}

\newcommand{\slqs}{\rSL^{(s)}_{q\hspace{-0.6mm}}(2,\C)}
\newcommand{\slpj}{\rSL^{(j)}_{p\hspace{-0.6mm}}(2,\C)}

\newcommand{\rSU}{\mathrm{SU}}

\newcommand{\rSL}{\mathrm{SL}}

\newcommand{\rSO}{\mathrm{SO}}

\newcommand{\rsu}{\mathrm{su}}

\newcommand{\rsl}{{\mathrm{sl}}}

\newcommand{\lra} {\longrightarrow}
\newcommand{\beq}{\begin{equation}}
\newcommand{\eeq}{\end{equation}}

\begin{document}

\begin{center}

{\Large \bf The $q$-linked complex Minkowski space,}

\bigskip

{\Large \bf its real forms and deformed isometry groups}

\bigskip

\centerline{ R. Fioresi$^\dagger$, E. Latini$^\star$, A. Marrani$^\sharp$}
\vskip 1 cm
$^\dagger${\sl Dipartimento di Matematica, Universit\`{a} di
Bologna\\Piazza di Porta S. Donato 5, I-40126 Bologna, Italy}\\
\texttt{rita.fioresi@UniBo.it}

\vskip 0.5 cm

$^\star$ {\sl Dipartimento di Matematica, Universit\`{a} di
Bologna\\Piazza di Porta S. Donato 5, I-40126 Bologna, Italy\\and INFN, Sez. di Bologna, \\
viale Berti Pichat 6/2, 40127 Bologna, Italy\\
 \texttt{emanuele.latini@UniBo.it}}

\vskip 0.5 cm

$^\sharp${\sl Museo Storico della Fisica e Centro Studi e
Ricerche ``Enrico Fermi" \\ Via Panisperna 89A, I-00184, Roma, Italy}
\\ \texttt{jazzphyzz@gmail.com}

 \end{center}

\bigskip

\medskip
\begin{abstract}
We establish duality between 
real forms of the quantum deformation of the 
4-dimensional orthogonal group studied by Fioresi et al. in \cite{flm}
and the classification work made by Borowiec
{\it et al.} in \cite{luk}. Classically these real forms are 
the isometry groups of 
$\mathbb{R}^4$ equipped with  Euclidean, Kleinian or Lorentzian metric. 
A general deformation, named $q$-linked, 
of each of these spaces is then constructed, 
together with the coaction of the corresponding
isometry group.

\end{abstract}

\section{Introduction}
One of the most intriguing problems in theoretical and mathematical physics 
is a consistent formulation of a quantum theory of gravity. 
The most relevant models are superstring theory
and loop quantum gravity; in both of them a minimal length or minimal
volume quanta occur, suggesting then that the
infinitely differentiable space-time manifold may be just a low energy 
picture, that should break down at very small scales. 
The Einstein theory of general relativity relates gravitational interaction to 
the geometry of space, 
thus it is natural to expect that,
for a complete and deep understanding of this problem, 
the classical notion of space-time should be superseded by some quantum notion \cite{Snyder} in which the minimal length plays a fundamental role and this task can be achieved in the framework of noncommutative
geometry. 
The main idea is then to use the duality ``spaces''-``algebras of 
functions'' and replace the commutative algebra of functions on a space 
by a noncommutative associative algebra regarded as the algebra of functions on
some noncommutative space \cite{frt,NC-1,NC-2,NC-3}. 
Physically, the algebra of 
``quantum space-time coordinate'' works as a regulator, with the 
Planck length being the minimal resolution length scale which
generates uncertainty relations among non commuting coordinates 
\cite{NC-4,ma2,NC-5,NC-6}.

This indicates 
 that the space-time symmetries must be
modified into quantum symmetries; mathematically this idea is 
captured  by  the theory of (non-commutative/cocommutative) Hopf
algebras, named, after Drinfeld, 
quantum groups \cite{quantum-groups}.

There are essentially two equivalent ways to approach quantum groups: 
via a deformation of the universal
enveloping algebra (QUEA) $\cU_q(\fg)$ of a semisimple Lie algebra $\fg$ or via
a quantization of the algebra of functions (QFA) $\mathrm{G}_q$
of a Lie group $\mathrm{G}$, 
(for more details see \cite{cp} Ch. 6). In this work,
our starting point are the deformations of  $\textrm{SO}(4)$ and its real forms compatible with a suitable quantum geometrical structure. In \cite{luk} the authors used the QUEA approach to classify all possible real forms of the quantum enveloping algebra $\cU_q(\mathrm{so}(4,\C))$ by means of the analysis of the $r$ matrix. Their idea was to view $\cU(\mathrm{ so}(4,\C))$ as the direct sum of two  $\cU(\rsl_2)$ (named chiral and antichiral sectors in analogy with the theory of spinors) and since the standard  and Jordanian deformations of $\cU(\rsl_2)$ and their real forms are well understood, the analysis of the coalgebra and star structure of $\cU_q(\mathrm{ so}(4,\C))$ follows almost naturally. Moreover in analyzing all possible classical $r$ matrices of the direct  sum of two $\cU_q(\rsl_2)$, in $\cite{luk}$ the authors showed that one can add a linking term 
mixing
the chiral and antichiral sector in the coalgebra
structure: we will denote this Hopf algebra by 
$\cU_{q,q'}(\rsl_2^{(s)}\stackrel{\ell}{\oplus}\rsl_2^{(s)})$.

\medskip
The main task of this paper is instead to construct the most general 
quantum deformation of the four dimensional Euclidean Klein and Minkowski 
spaces viewed as quantum homogeneous spaces
\footnote{An alternative approach to directly obtain the real forms by using (split)division algebras was discussed also in \cite{FL-1,flm2}.}. 
In order to achieve this task the QFA approach turns out to be more convenient, since it is easier to understand the comodule structures. We thus construct the duality between the Hopf star algebras classified in  \cite{luk} and certain generalization of the QFA of $\mathrm{SO}(4)$ discussed in \cite{flm}. The explicit realization of this duality is also motivated by the aim of developing a differential calculus on the deformed Minkowski space, viewed as a quantum homogeneous space, along the line of \cite{Zampini}. This will allow us to define the complex $q$-linked Minkowski space as a special 
comodule of the complex Lorentz group 
$ \mathrm{SL}_q^{(s)}(2,\C)\stackrel{\ell}{\oplus}  \mathrm{SL}_{q}^{(s)}(2,\C)$. 
 Technically this is realized, along the line of \cite{fl,cfln,fi2,fi3,fi4} 
(see also \cite{flv,cfl,flln} for a generalization to the super setting) and
\cite{Manin2, Sudbery}. 
We also speculate on the interpretation of the "new" parameter $\ell$ ,
showing that the product of two $\ell$ commuting (or linked) 
quantum planes (or quantum spinors in the physical language \cite{qspin}) 
reproduces the deformed algebra of the complex $q$-linked Minkowski space.
This interpretation nicely fits with a recent work \cite{Landi}, 
where the non commutative product of spheres and planes is studied.
Notice 
that our work is somehow a generalization of the construction introduced by 
Truini and Varadarajan in \cite{tv}, where they first 
introduced the link parameter, denoted here by $\ell$.
In the end, we analyze the compatible real structures 
of the complex orthogonal quantum group in four dimensions
with the $q$-linked quantum Klein, Euclidean and Minkowski spaces.

\medskip
The organization of this paper is as follows.

\medskip
In Sections \ref{prelim} and 3  we describe how to
obtain a quantum deformation of the complex
special linear group, by means of the QUEA and QFA approach. We then establish
a duality between the complex standard and Jordanian QUEA and the corresponding
QFA. The material of this section is generally known, however, since the
presence of several conventions in the literature, for the
reader's convenience, we have briefly recapped the main results we need.

In Section 4 we analyze the real forms 
of the quantum orthogonal group compatible with the duality discussed above.

Section 5 is then devoted to the construction and definition of 
interesting comodules of those real structures whose 
classical limit (parameters $q$ and $\ell$ specialized to one) 
coincide with Euclidean Klein and Minkowski spaces.

\bigskip
{\bf Acknowledgements}. We wish to thank 
P. Truini, F. Gavarini, F. Bonechi, A. Zampini
and V.S. Varadarajan for helpful comments.

\section{Preliminaries}\label{prelim}

In this section, we establish our notation and provide the
two equivalent settings of QUEA and QFA deformations in the explicit example of 
$\fg=\rsl_2(\C)$ and $\mathrm{G}=\rSL_2(\C)$, together with
the duality between QUEA's and QFA's, that we need
for our applications.
This material can be found in \cite{cp,kassel,ks}.

The approach to the theory of quantum
deformations via the QUEA
is the most used in the literature (see \cite{cp} Ch. 6 and references therein).
We can obtain deformations of the enveloping algebra $\cU(\rsl_2)$
by a standard or a Jordanian $r$ matrix yielding two different Hopf algebras 
structures, denoted respectively by $\cU_q^{(s)}\hspace{-0.5mm}(\rsl_2)$ 
and $\cU_p^{(j)}\hspace{-0.5mm}(\rsl_2)$. 

\medskip
{\bf Standard deformation (QUEA):} $\Uqs$. 
The algebra $\Uqs$
is the complex associative
algebra generated by $e_\pm$, $q^{\pm \bh}$ and
subject to the relations 
\beq\label{dj-rel}
\begin{array}{c} 
[e_+,e_-]=\frac{q^{2\bh}-q^{-2\bh}}{q-q^{-1}},\qquad
q^\bh e_{\pm}=q^{\pm 1}e_{\pm}q^\bh,\qquad
q^{\bh}q^{-\bh}=1
\end{array}
\eeq
with comultiplication, counit  given by:
\beq\label{dj-rel-coalgebra}
\begin{array}{rclcrcl}
\Delta(e_+)&=&e_{+} \otimes q^{\bh}+ q^{-\bh} \otimes e_{+}, &&&&\\&&&&\epsilon(e_\pm)&=&0, \\
\Delta(e_-)&=&e_{-} \otimes q^{\bh}+ q^{-\bh} \otimes e_{-}, &&&&\\
&&&&\epsilon(q^{\pm \bh} )&=&1  \\
\Delta(q^{\bh})&=&q^{\bh}\oplus q^{\bh}&&&&
\end{array}
\eeq
and antipode 
\beq\label{dj-rel-antipode}
S(e_{\pm })=-q^{\pm 1}e_{\pm},\quad S(q^{\pm \bh})=q^{\mp \bh}
\eeq


{\bf Jordanian deformation (QUEA): $\Uj$}. 
The algebra $\Uj$ is the
complex associative algebra generated by $E$, $F$, $H$
satisfying the 
relations:
\beq\label{j-rel}
[H,E]=E\,,\,\,\,\,\,\,\,\,\,\,\,\,[H,F]=F\,,\,\,\,\,\,\,\,\,\,\,\,\,[E,F]=2H
\eeq
In this case the $r$ matrix 
is triangular
(see \cite{luk,Ogie,Herranz}).
The Hopf algebra structure is given by:
$$
\begin{array}{rcl}
\Delta(E)&=&E \otimes e^{\sigma}+ 1 \otimes E\\[1mm]
\Delta(H)&=&H \otimes e^{-\sigma}+ 1 \otimes H\\[1mm]\Delta(F)&=&F \otimes 
e^{-\sigma}+1\otimes F-p H\otimes He^{-\sigma}-4p^2H(H+1)\otimes Ee^{-2\sigma}
\end{array}
$$
and 

$$
\begin{array}{rcl}
S(E)&=&-Ee^{-\sigma}\\
S(H)&=&-He^{-\sigma}\\
S(F)&=&-Fe^{\sigma}-4pH^2e^\sigma+4pH(H+1)Ee^{\sigma}
\end{array}
$$
where we have defined
$$
\sigma=\ln(1-2p E)
$$
In order to get a better expression for the coalgebra structure
we choose a different set of generators: 
$$
\begin{array}{c} 
t=e^{\frac \sigma 2},\quad
h=2He^{-\frac \sigma 2},\quad
y=e^{-\frac \sigma 2}(F+2pH)+\frac p 4 \sinh\frac \sigma 2
\end{array}
$$
These generators satisfy the relations:
$$
\begin{array}{c} 
[h,t]=t^2-1,\quad
[y,t]=-\frac p 2 (ht+th),\quad
[h,y]=-\frac 1 2(yt+ty+yt^{-1}+t^{-1}y).
\end{array}
$$
The coalgebra structure becomes
$$
\begin{array}{lr} 
\Delta(h)=h \otimes t+ t^{-1} \otimes h,\quad
\Delta(y)=y \otimes t+ t^{-1} \otimes y, \qquad
\\ \\
\Delta(t)=t \otimes t, \qquad 
\epsilon(h)=0=\epsilon(y),\quad\epsilon(t )=1=\epsilon(t^{-1})
\end{array}
$$
and
$$
\begin{array}{c} 
S(h)=-tht^{-1}, \quad
S(y)=-tyt^{-1}, \quad
S(t)=te^{-1}.
\end{array}
$$
     
\section{The function algebra approach}

We now construct the deformations of the function algebra of
 $\rsl_2(\C)$ corresponding to the standard and Jordanian QUEA
$\Uqs$, $\Uj$. 

\medskip
{\bf Standard deformation (QFA): $\slqs$}.
We define $\slqs$, the standard  deformation
of the function algebra of the algebraic group $\rSL_2(\C)$, 
as generated over $\C[q,q^{-1}]$ by $a$, $b$, $c$, $d$
subject to the relations, called the \textit{Manin relations}:
\beq\label{manin-rels}
\begin{array}{rclcrcl}
ab&=&q^{-1}ba, &\quad\qquad\quad&  ac&=&q^{-1}ca\\
 bd&=&q^{-1}db, &\quad\qquad\quad& cd&=&q^{-1}dc \\ 
ad-da&=&(q^{-1}-q)bc, &\quad\qquad\quad& bc&=&cb \\
\end{array}
\eeq
and
\beq\label{qdet}
ad-q^{-1}bc=1
\eeq
(See \cite{ma1} and \cite{ks}.
Warning: for \cite{ks} $q$ is replaced with $q^{-1}$).
The parameter $q$ may be interpreted both as an indeterminate
or as $q=e^h$, $h \in\C$. 
We retrieve the commutative function
algebra of $\rSL_2(\C)$ when we specialize $q$ to $1$.
The last relation (\ref{qdet}) imposes a constraint
on the so called \textit{quantum determinant} $ad-q^{-1}bc$ (see also
\cite{fi2, fi3, fi4} for more details).

$\slqs$ is endowed with an Hopf algebra structure by setting:
$$
\begin{array}{rcl}
\Delta \begin{pmatrix} a&b\\c&d\end{pmatrix}&=&\begin{pmatrix} a&b\\c&d\end{pmatrix}\otimes \begin{pmatrix} a&b\\c&d\end{pmatrix}\\[5mm]
\epsilon\begin{pmatrix} a&b\\c&d\end{pmatrix}&=&\begin{pmatrix} 1&0\\0&1\end{pmatrix}\\[5mm]
S\begin{pmatrix} a&b\\c&d\end{pmatrix}&=&\begin{pmatrix} d&-qb\\-q^{-1}c&a\end{pmatrix}
\end{array}
$$
where we use the common notation (see \cite{ma1, ks}), in which we 
organize the generators $a$, $b$, $c$, $d$ in
matrix form.

\medskip
{\bf Jordanian deformation (QFA):} $\slpj$.
Following \cite{Ohn} we define $\slpj$, the Jordanian  deformation
of the function algebra of $\rSL_2(\C)$, 
as generated over $\C[p,p^{-1}]$ by $\al$, $\be$, $\ga$, $\de$ subject
to the relations:
$$
\begin{array}{rclcrcl}
[\al,\be]&=&p\al^2-p, &\quad\qquad\quad&  [\al,\ga]&=&-p\ga^2\\[2mm]
 [\be,\de]&=&p-p\de^2, &\quad\qquad\quad& [\ga,\de]&=&p\ga^2 \\[2mm]
[\al,\de]&=&p(\al-\de)\ga, &\quad\qquad\quad& [\be,\ga]&=&-p\al\ga-p\ga\de
\end{array}
$$
and
$$
\al\de-\be\ga-p\al\ga=1
$$
Notice
that $p$ plays the role of $h$ in the standard QFA, $\slqs$
 and can be interpreted
either as a number or as an indeterminate. We retrieve the commutative
function algebra of  $\rSL_2(\C)$ by setting $p=0$.

This algebra is endowed with an Hopf algebra structure by 
defining $\Delta$,  $\epsilon$, $S$ as (see for example \cite{Manin,Herranz}):
$$
\begin{array}{rcl}
\Delta \begin{pmatrix} \al&\be\\\ga&\de\end{pmatrix}&=&\begin{pmatrix} \al&\be\\\ga&\de\end{pmatrix}\otimes\begin{pmatrix} \al&\be\\\ga&\de\end{pmatrix}\\[5mm]
\epsilon\begin{pmatrix} \al&\be\\\ga&\de\end{pmatrix}&=&\begin{pmatrix} 1&0\\0&1\end{pmatrix}\\[5mm]
S\begin{pmatrix} \al&\be\\\ga&\de\end{pmatrix}&=&\begin{pmatrix} \de-p\ga&-\be+p(\al-\de)+p^2\ga\\-\ga&\al+p\ga\end{pmatrix}
\end{array}
$$
where again we use the matrix notation for the generators.

\subsection{The duality between QUEA and QFA}

An element of the universal enveloping algebra of a Lie algebra $\fg$
is naturally viewed as a differential operator on the algebra of function
on the Lie group G whose Lie algebra is $\fg$. This establishes
a perfect duality between the universal enveloping algebra and the algebra
of functions. This pairing naturally generalizes to give
a non degenerate pairing between QUEA's and QFA's.

\begin{proposition}\label{dual-prop}
There is a family of non degenerate dual pairings 
$\langle \bullet,\bullet \rangle$, parametrized by a constant 
$\varepsilon$, given by
$$
\langle q^{\bh}
,a\rangle=q^{\frac 1 2}, \quad
\langle q^{\bh},d\rangle=q^{-\frac 1 2}, \quad \langle e_+,b\rangle=
q^{\varepsilon+\frac 1 2}
,\qquad \langle e_-,c\rangle=q^{-\varepsilon-\frac 1 2} 
$$
yielding a perfect duality between $\Us$ and $\Ss$.
\end{proposition}

\begin{proof}
See \cite{ks} pg 120 and \cite{kassel} pg 152 and explicit calculation.
\end{proof}
Notice that this duality may be extended to higher dimensions.\\[3mm]

A similar duality is found also for the nonstandard case.

\begin{proposition}\label{dual-prop}
The dual pairing 
 $$
 \begin{array}{rclcrclcrcl}
\langle h,\al\rangle&=&1, &\langle h,\de\rangle&=&-1,&\langle y,\ga\rangle&=&1,\\[2mm] 
\langle t,\al\rangle&=&1, &\langle t,\be\rangle&=&p,&\langle t,\de\rangle&=&1, \\[2mm]
\langle t^{-1},\al\rangle&=&1, &\langle t^{-1},\be\rangle&=&-p,&\langle t^{-1},\de\rangle&=&1 
\end{array}
$$
produces a duality between $\Uj$ and $\Sj$.
\end{proposition}

\begin{proof}
This is an explicit calculation, see for example  \cite{Ohn}.
\end{proof}

\section{Real forms of QUEA's and QFA's} \label{realforms}

In this section we examine the real forms of the QUEA's and QFA's
discussed in the previous section, together with their corresponding
dualities.

\subsection{Real forms of quantum groups} 

A real form of a quantum group ({\it e.g.} a QUEA or QFA),
is understood as a complex Hopf algebra together
with a \textit{star structure} (see \cite{cp} 4.1) and \cite{ks} Ch. 1).

\begin{definition}
An \textit{Hopf star algebra} is an complex Hopf algebra $\mathcal{H}$  
together with a \textit{star structure},
that is a antilinear unitary antiautomorphism $*$,
which is also an involution and satisfies:
$$ 
\Delta(a^*)=\Delta(a)^*, \quad S(S(a^*)^*)=a, \quad \ep(a^*)=\overline{\ep(a)}
$$
We will denote such Hopf star algebra by the pair $(\mathcal{H},*)$ and
by an abuse of terminology we shall call it a \textit{real form} of $\cH$.
Also, with a slight abuse of terminology, we will call star structure
any antilinear antiinvolution on a complex algebra. 
\end{definition}

This definition agrees with the one in \cite{ks} pg 20, though
it must be noticed that the last two properties are a consequence
of the other requirements (see Prop. 10 pg 20 in \cite{ks}).

\medskip
The definition of star structure on 
an Hopf algebra $\cH$ involves the coalgebra structure of $\cH$,
hence it may well be that isomorphic real forms of a certain Lie algebra
$\fg$ yield different ({\it i.e.} nonisomorphic) real forms of 
the corresponding universal enveloping algebra $\cU(\fg)$; 
in other words the same complex universal enveloping algebra, viewed
as an Hopf algebra, may
have two nonisomorphic star structures, corresponding to the same
real form of $\fg$. This will then give non isomorphic QUEA
star Hopf algebras, corresponding to the same real form of $\fg$.

For example, 
in 3.1 \cite{luk}, the authors present four non isomorphic bialgebras
which are real forms of $\rsl_2(\C)$: three are the
classical standard bialgebras $\rsu(2)$,  $\rsu(1,1)$, $\rsl_2(\R)$,
while the fourth one is a non standard bialgebra structure on 
$\rsl(2)$. 
Notice that, 
while $\rsu(1,1) \cong \rsl_2(\R)$ as Lie algebras, they 
are non isomorphic as bialgebras (for the bialgebra structure 
see \cite{luk}). Consequently
the star structures on their
universal enveloping algebras and the deformations
 QUEA constructed in Sec. 4 of \cite{luk},
corresponding to these bialgebra structures, will yield non isomorphic
Hopf star algebras.

\begin{definition}
We say that two Hopf star algebras $(\mathcal{H},*)$ and ($\mathcal{F}$, $\star$) are in perfect duality if there exists a
pairing $\langle \bullet, \bullet\rangle$ 
yielding a perfect duality between $\mathcal{H}$ and $\mathcal{F}$ and such that
$$
\langle x^*,f\rangle=\overline{\langle x,S(f)^\star\rangle}, \qquad
\forall x\in \mathcal{H}, \quad\forall f\in\mathcal{F}.
$$  
\end{definition}

\label{stars}

\subsection{Real forms of Standard and Jordanian QUEA's
and QFA's}\label{realf-sec}

In the language of the previous subsection, we shall list several
star structures for the QUEA deformations $\Uqs$, $\Uj$
and their corresponding QFA deformations $\slqs$, $\slpj$, together
with their dualities, introduced in Sec. \ref{prelim}. We leave to
the reader the tedious checks involved in all of our statements.

\medskip
{\bf Standard deformations}. 
For each real form of $\rsl_2(\C)$,
we list the real QUEA and QFA, as Hopf star algebras,
and the perfect duality between them.
  
\begin{itemize}
\item $\rSU(2)$, the special unitary group.
$$
\begin{array}{c}
\hbox{QUEA} \quad
\cU_q^{(s)}\hspace{-0.5mm}(\rsu(2)):=(\Us,*_{\hspace{-0.8mm}
\textrm{\tiny{e}}})\,\,
\textrm{with} \,\,
(q^{\bh})^{\se}=q^{\bh}, \quad, e_{\pm}^{\se}=e_{\mp}, \, q\in\mathbb{R}\\ \\
\hbox{QFA} \quad
\rSU^{\hspace{-0.2mm}(s)}_q\hspace{-0.6mm}(2):=(\Ss,\sre)\,\,
\textrm{with} \,\,
\begin{pmatrix} a&b\\c&d\end{pmatrix}^{\sre}=
S\begin{pmatrix} a&c\\b&d\end{pmatrix},\quad q\in\mathbb{R}
\end{array}
$$
Dual pairing 
$$
\langle q^{\bh}
,a\rangle=q^{\frac 1 2}, \quad
\langle q^{\bh},d\rangle=q^{-\frac 1 2}, \quad \langle e_+,b\rangle=1,\qquad \langle e_-,c\rangle=1
$$

\item $\rSU(1,1)$.
$$
\cU_q^{(s)}\hspace{-0.5mm}(\rsu(1,1)):=(\Us,\sk)\,,\,\,\,\,\,\,
\textrm{with}\,\,\,\,
(q^{\bh})^{\sk}=q^{\bh}\,,\,\,\,\,\,\, e_{\pm}^{\sk}=-e_{\mp}\,;\,\,\,\,\,\,\,\,\,q\in\mathbb{R}
$$
$$\rSU^{\hspace{-0.2mm}(s)}_q\hspace{-0.6mm}(1,1):=(\Ss,\srk)\,,\,\,\,\,\,\,\textrm{with} 
\begin{pmatrix} a&b\\c&d\end{pmatrix}^{\srk}=\begin{pmatrix} d&qc\\q^{-1}b&a\end{pmatrix};\,\,\,\,\,\,\,\,\,q\in\mathbb{R}
$$
Dual pairing 
$$
\langle q^{\bh}
,a\rangle=q^{\frac 1 2}, \quad
\langle q^{\bh},d\rangle=q^{-\frac 1 2}, \quad \langle e_+,b\rangle=1,\qquad \langle e_-,c\rangle=1
$$

\item $\rSL_2(\R)$ the real special linear group. 
$$\cU_q^{(s)}\hspace{-0.5mm}((\rsl_2(\mathbb{R})):=(\Us,\sts)\,,\,\,\,\,\,\,
\textrm{with}\,\,\,\,
(q^{\bh})^{\sts}=q^{\bh}\,,\,\,\,\,\,\, e_{\pm}^{\sts}=-e_{\pm}\,;\,\,\,\,\,\,\,\,\,q\bar{q}=1
$$
$$\rSL^{\hspace{-0.2mm}(s)}_q\hspace{-0.6mm}(2,\mathbb{R}):=(\Ss,\srs)\,,\,\,\,\,\,\,\textrm{with} 
\begin{pmatrix} a&b\\c&d\end{pmatrix}^{\srs}=\begin{pmatrix} a&b\\c&d\end{pmatrix}\,;\,\,\,\,\,\,\,\,\,q\bar{q}=1
$$
Dual pairing 
$$
\langle q^{\bh}
,a\rangle=q^{\frac 1 2}, \quad
\langle q^{\bh},d\rangle=q^{-\frac 1 2}, \quad \langle e_+,b\rangle=q^{\frac 1 2},\qquad \langle e_-,c\rangle=q^{-\frac 1 2}
$$
\end{itemize}
\begin{remark} Notice that in this case we could equivalently 
define the involution to be multiplicative instead of being 
antimultiplicative by forcing $q$ to be real. 
\end{remark}

\medskip
{\bf Jordanian deformation}. 
In this case there is only one real structure 
and it corresponds to the real form $\rsl(2)(\R)$. 
$$
\begin{array}{rcl}
\cU_p^{(j)}\hspace{-0.5mm}((\rsl_2(\R))&:=&(\Uj,*)\,,\,
\textrm{with}\\[3mm]
(t)^*&=&t\,,\,\,\, h^*=-h\,,\,\,\, y^*=-y\,,\,\,\,\bar{p}=-p
\end{array}
$$
$$
\begin{array}{rcl}\rSL^{\hspace{-0.2mm}(j)}_p\hspace{-0.6mm}(2,\mathbb{R})&:=&(\Sj,\star)\,,\,\textrm{with} \\[3mm]
\begin{pmatrix} \al&\be\\ \ga&\de\end{pmatrix}^\star&=&\begin{pmatrix} \al+p\ga&\be-p\al-p^2\ga+p\de  \\\ga&\de-p\ga\end{pmatrix}=S\begin{pmatrix}\de&-\be\\ -\ga&\al\end{pmatrix}\,;\,\,\,\,\,\,\,\,\,\bar{p}=-p
\end{array}
$$
Dual pairing: see Prop. \ref{dual-prop}.

\section{The $q$-linked complex 4-d special orthogonal group, and its real forms}

We are now interested in defining the most general 
deformation of the complex 4-d orthogonal group $\mathrm{SO}(4,\C)$, 
together with its real forms.

\label{qsog}\subsection{Quantum special
orthogonal group}

We proceed as we did in the previous sections and consider
quantum deformations of the special
orthogonal group in four dimensions 
through the QUEA and QFA approaches establishing
a duality between them.

\medskip
{\bf QUEA deformations.}
The identification $\mathrm{so}(4,\C) \cong 
\mathrm{sl}_2(\C) \oplus \mathrm{sl}_2(\C)$ allows us to define
immediately the following three quantum deformations:
\beq\label{def-so4}
\begin{array}{c} \cU_{q,q'}(\rsl_2^{(s)} \oplus \rsl_2^{(s)}), \qquad
\cU_{q,q'}(\rsl_2^{(s)} \oplus \rsl_2^{(j)}), \qquad
\cU_{q,q'}(\rsl_2^{(j)} \oplus \rsl_2^{(j)}).
\end{array}
\eeq
We refer to the left and right summand respectively as 
{\sl chiral and antichiral sectors}. In our notation, 
$\cU_{q,q'}(\rsl_2^{(s)} \oplus \rsl_2^{(s)})$ is the Hopf algebra generated by 
$e_\pm$, $q^{\pm \bh}$ and $e'_{\pm}$, $q'^{\pm \bh'}$ in the chiral and antichiral sector respectively, subject to the relations (\ref{dj-rel}) for each copy,
where the Hopf algebra structure is given by two copies of 
(\ref{dj-rel-coalgebra}), (\ref{dj-rel-antipode}) and the algebraic structure is simply obtained by allowing the elements of the two sectors to commute. The definition of
$\cU_{q,q'}(\rsl_2^{(s)} \oplus \rsl_2^{(j)})$ and
$\cU_{q,q'}(\rsl_2^{(j)} \oplus \rsl_2^{(j)})$ is similar, taking 
(\ref{j-rel}) instead of (\ref{dj-rel})
when the index $(j)$ occurs (and also the corresponding Hopf algebra
jordanian structure).

This construction is not the most general one. In fact, following the spirit of 
\cite{tv}, one can try to deform the \emph{link} between the chiral
and the antichiral sectors. This essentially amounts to deform the coalgebra structure 
(equivalently the algebra structure in the QFA approach) 
by mixing the chiral and antichiral sectors. 
This was already noticed in \cite{luk}, 
where the authors remarked that in order to exhaust all possible bialgebra 
structures it is necessary to consider the mixed terms in the 
classical $r$-matrix, that is those terms belonging to 
$\mathrm{sl}_2(\C)\wedge \mathrm{sl}_2(\C)$. 
These terms produce 
two Jordanian bialgebras $\Uj$ intertwined by an Abelian twist 
and the deformation of $\mathrm{so}_4(\C)$ twisted by the Belavin-Drinfeld triple.

\medskip
In the following we will focus our attention on the so called 
standard-standard case $\cU_{q,q'}(\rsl_2^{(s)} \oplus \rsl_2^{(s)})$. 
(refer to \cite{luk}).

\begin{definition}
We define the {\it q-linked complex 4-d special orthogonal group} (QUEA)
as the complex $\cU_{q,q'}(\rsl_2^{(s)} \stackrel{\ell}
\oplus \rsl_2^{(s)})$ associative algebra generated by
$e_\pm$, $q^{\pm \bh}$ and $e'_{\pm}$, $q'^{\pm \bh'}$ subject to the 
algebraic relations (see (\ref{dj-rel})). 
The coalgebra structure given by:
$$
\begin{array}{rcl}
\Delta(e_\pm)&=&e_{\pm} \otimes q^{\bh}\ell^{\pm \bh'}+ q^{-\bh}\ell^{\mp \bh'} \otimes e_{\pm}\\[3mm]
\Delta(e'_\pm)&=&e'_{\pm} \otimes q'^{\bh'}\ell^{\mp \bh}+ q'^{-\bh'}\ell^{\pm \bh} \otimes e'_{\pm}\\[3mm]
\Delta(q^{\pm \bh})&=&q^{\pm \bh} \otimes q^{\pm \bh}\\[3mm]
\Delta(q'^{\pm \bh'})&=&q'^{\pm \bh'} \otimes q'^{\pm \bh'}
\end{array}
$$
As for the antipode, we have (\ref{dj-rel-antipode}) for each of the chiral
and antichiral sectors, that is: 
$$
S(e_{\pm })=-q^{\pm 1}e_{\pm},\quad S(q^{\pm \bh})=q^{\mp \bh}, \qquad
S(e'_{\pm })=-q'^{\pm 1}e'_{\pm},\quad S(q'^{\pm \bh})=q'^{\mp \bh'}
$$
\end{definition}
Notice that the link parameter $\ell$ occurs only in the deformation
of the coalgebra structure.

\bigskip
{\bf QFA deformation}. We now take a dual approach and provide
a deformation of the function algebra of the special orthogonal group
in four dimensions.
Inspired by \cite{tv}, we give the following definition.

\begin{definition}
We define \textit{$q$-linked complex 4-d special orthogonal group} (QFA) 
denoted by $\rSL_q^{\hspace{-0.6mm}(s)}
\hspace{-0.6mm}(2,\C)\stackrel{\ell}{\oplus} 
\rSL_{1/q'}^{\hspace{-0.6mm}(s)}\hspace{-0.4mm}(2,\C)$ as the associative
complex algebra generated by the indeterminates $a$, $b$, $c$, $d$
and $a'$, $b'$, $c'$, $d'$ subject to the Manin relations (see (\ref{manin-rels}))
and to the extra {\sl link} relations among the two sectors:
\beq\label{linkr}
\begin{array}{rclccrcl}
b'a&=&\ell ab'&\,\,\,\,\,&ac'&=&\ell c'a\\[2mm]
a'b&=&\ell ba' &&bd'&=&\ell d'b\\[2mm]
ca'&=&\ell a'c& &d'c&=&\ell cd'\\[2mm]
db'&=&\ell b'd&&c'd&=&\ell  dc'\\[2mm]
aa'&=& a'a&\,\,\,\,\,&ad'&=&d'a\\[2mm]
bb'&=& b'b &&bc'&=& c'b\\[2mm]
cb'&=& b'c& &cc'&=&c'c\\[2mm]
da'&=&a'd&&dd'&=&  d'd
\end{array}
\eeq

When $q=\frac{1}{q'}$ we will also refer to it as the $q$-linked complex Lorentz group or complex Lorentz quantum group.
\begin{remark}We observe that the left and chiral determinants $\det:=ad-q^{-1}bc$ and $\det':=a'd'-q'b'c'$ are again central, and that consistency check is crucial.
\end{remark}

\end{definition}
With a common abuse of notation we denote by:
$$
(x_{ij})=\renewcommand\arraystretch{0.5}\begin{pmatrix} a&b\\c&d\end{pmatrix}
$$ 
the set of the 4 generators $a$, $b$, $c$, $d$ and we call $(x_{ij})$ a $q$
matrix, meaning that such generators satisfy the Manin relations with 
parameter $q$. Similarly we denote by:
$$
(y_{ij})=\renewcommand\arraystretch{0.5}\begin{pmatrix} a'&b'\\c'&d'\end{pmatrix}
$$ 
and we notice that $(y_{ij})$ is a $(q')^{-1}$ matrix.

\bigskip
As expected the two approaches QUEA and QFA to the
deformations of the complex special orthogonal group yield Hopf algebras in 
perfect duality.

\begin{proposition}\label{dual}
The dual paring 
$$
\begin{array}{rclccrclcrclcrcl}
\langle q^{\bh}
,a\rangle&=&q^{\frac 1 2},& \quad
&\langle q^{\bh},d\rangle&=&q^{-\frac 1 2},&\,& \langle e_+,b\rangle&=&1,&\,& \langle e_-,c\rangle&=&1,\\[5mm]
\langle q'^{\bh'}
,a'\rangle&=&q'^{-\frac 1 2},& \,
&\langle q'^{\bh'},d'\rangle&=&q'^{\frac 1 2},& \,& \langle e'_+,c'\rangle&=&1,&\,& \langle e'_-,b'\rangle&=&1,\\[5mm]
\langle q^{\bh}
,a'\rangle&=&1,& \quad
&\langle q^{\bh},d'\rangle&=& 1,&\,&\langle q'^{\bh'}
,a\rangle&=&1 ,& \,
&\langle q'^{\bh'},d\rangle&=&1\\[5mm]
\langle \ell^{\bh}
,a\rangle&=&\ell^{\frac 1 2},& \quad
&\langle \ell^{\bh},d\rangle&=&\ell^{-\frac 1 2},&\,& \langle \ell^{\bh}
,a'\rangle&=&\ell^{-\frac 1 2},& \quad
&\langle \ell^{\bh},d'\rangle&=&\ell^{\frac 1 2}\\[5mm]
\langle \ell^{\bh}
,a'\rangle&=&1,& \quad
&\langle \ell^{\bh},d'\rangle&=& 1,&\,&\langle \ell^{\bh'}
,a\rangle&=&1 ,& \,
&\langle \ell^{\bh'},d\rangle&=&1
\end{array}
$$
yields a perfect duality between the 
Hopf algebras $\cU_{q,q'}(\rsl_2^{(s)} \stackrel{\ell}
\oplus \rsl_2^{(s)})$ 
and $\rSL_q^{\hspace{-0.6mm}(s)}\hspace{-0.6mm}(2,\C)\stackrel{\ell}{\oplus} 
\rSL_{1/q'}^{\hspace{-0.6mm}(s)}\hspace{-0.4mm}(2,\C)$.
\end{proposition}

\begin{proof}
We first show that the algebraic and coalgebraic structures of $\cU_{q,q'}(\rsl_2^{(s)} \stackrel{\ell}
\oplus \rsl_2^{(s)})$ 
and $\rSL_q^{\hspace{-0.6mm}(s)}\hspace{-0.6mm}(2,\C)\stackrel{\ell}{\oplus} 
\rSL_{1/q'}^{\hspace{-0.6mm}(s)}\hspace{-0.4mm}(2,\C)$ are consistent with the requirement 

$$
\langle x, fg\rangle=\sum \langle x_{(1)},f\rangle \langle x_{(2)},g\rangle
$$
and
$$
\langle xy,f\rangle=\langle x,f_{(1)}\rangle \langle y,f_{(2)}\rangle
$$
for every $x,y\in\cU_{q,q'}(\rsl_2^{(s)} \stackrel{\ell}
\oplus \rsl_2^{(s)})$ and $f,g\in\rSL_q^{\hspace{-0.6mm}(s)}\hspace{-0.6mm}(2,\C)\stackrel{\ell}{\oplus} 
\rSL_{1/q'}^{\hspace{-0.6mm}(s)}\hspace{-0.4mm}(2,\C)$, where in the  Sweedler notation we have $\Delta(x)=\sum x_{(1)}\otimes x_{(2)}$. The most involved case in when we look at the algebra structure on the QFA side and the coalgebra 
structure on the QUEA one, because they are modified by the link parameter. Let us start looking at the chiral (equivalently antichiral) sector only; note that

$$
\langle q^{\bh}\ell^{\pm \bh'},f\rangle=\langle q^{\bh},f\rangle
$$
thus, in this case, the extra terms proportional to $\ell$ in the coproduct  of $e_{\pm}$ do not contribute and the proof of our claim naturally follows by inspection. Observe then that since $b$ and $c$ have vanishing pairing with the antichiral generators on the enveloping algebra side, the other relations are trivial. The situation is obviously symmetric thus we are left to prove the consistency of the link relations (\ref{linkr}) with the pairing proposed. This must be done case by case, as sample calculation we check the case of the link relation $ba'$. The only non trivial case is when we pair it with the term $e_+q^{n\bh}$ for any $n\in\mathbb{Z}$. In particular one gets
$$
\langle e_+q^{n\bh},ba'\rangle=\langle e_+q^{n\bf h},b\rangle \langle q^{(n+1)\bh}\ell^{\bf h'},a'\rangle=q^{-\frac n 2}\ell^{-\frac 1 2}
$$
while
$$
\langle e_+q^{n\bh},a'b\rangle=\langle q^{-\bh}\ell^{-\bf h'}q^{n\bh},a'\rangle\langle e_+q^{n\bf h},b\rangle=\ell^{\frac 1 2}q^{-\frac n 2 }
$$ 
proving then that $a'b=\ell ba'$ as claimed. The paring of the form $\langle xy,f\rangle$ trivially satisfies the required condition when we work on the chiral or antichiral sector only. When on the left hand side we have a mixed term, the only non vanishing pairing is when $xy=q^{n\bh}q'^{m\bh'}$ for some $n, m\in\mathbb{Z}$ and $f=a \,\,\textrm{or} \,\,d$ (equivalently $f=a' \,\,\textrm{or} \,\,d'$). In those cases one can check that the duality holds, and using the same arguments the proof can be extended to the whole algebra. We have then constructed a duality of algebras, but we have indeed something more; in fact it is easy to check on generators that 
$$
\langle S(x),f\rangle=\langle x,S(f)\rangle
$$
thus we got a duality of Hopf algebras. Note then that so far we have used only the Manin and link relations; we observe that the chiral and antichiral determinants $\det$ and $\det'$ commute and according to our rules, we have
$$
\langle e'_{\pm}, \det \det\nolimits{'} \rangle=\langle e_{\pm}, \det \det\nolimits{'} \rangle=0=\epsilon(e_{\pm})=\epsilon(e'_{\pm})
$$
while
$$
\langle q'^{\bf h'}, \det \det\nolimits{'}\rangle=\langle q^{\bf h}, \det \det\nolimits{'}\rangle=1=\epsilon(q^{\bf h})=\epsilon(q'^{\bf h'})
$$
implying that the duality is perfect only when $\det=1=\det'$, that is the case of $\rSL_q^{\hspace{-0.6mm}(s)}\hspace{-0.6mm}(2,\C)\stackrel{\ell}{\oplus} 
\rSL_{1/q'}^{\hspace{-0.6mm}(s)}\hspace{-0.4mm}(2,\C)$.

\end{proof}

\subsection{Real forms}
We are now interested in the real forms the quantum complex Lorentz group
defined through the QUEA and QFA approaches on the previous section.

\medskip
Let us examine first the case in which the star structure on the quantum complex
Lorentz group (both in the QUEA and QFA approaches)
does not mix the chiral and antichiral sectors. Let $\sre$ denote the star
structure on $\rSL_q^{\hspace{-0.6mm}(s)}(2,\C)$,
whose fixed points correspond to the 
real Euclidean quantum group and $\srk$ denote the star
structure on $\rSL_q^{\hspace{-0.6mm}(s)}(2,\C)$, 
whose fixed points correspond to the 
real Klein quantum group 
(see \cite{flm} for more details on such star structures).
We have six possibilities: the combination $(\sre\oplus \sre)$ 
yields the deformation of the Euclidean group $\rSO_{q,q'}^{\,\,\,\ell}(4)$, 
the triplet $(\srs\oplus \srs)$, $(\srk\oplus \srk)$ and $(\srk\oplus \srs)$ 
yields three different deformations of the Klein group we denote $\rSO_{q,q'}^{\,\,\,\ell}(2,2)$, $\rSO_{\,q,q'}^{'\,\,\,\ell}(2,2)$ and $\rSO_{\,q,q'}^{''\,\,\ell}(2,2)$ respectively. In the end 
the pairs $(\sre\oplus \srs)$ and $(\sre\oplus \srk)$ are related to the 
so called {\sl quaternionic symmetry}; we decide to disregard those
last two real forms since, 
even in the classical case, they do not correspond to interesting comodule 
structures with a natural geometrical interpretation. 

There is however an extra real form of the complex
Lorentz group $\rSL_q^{\hspace{-0.6mm}(s)}\hspace{-0.6mm}(2,\C)\stackrel{\ell}{\oplus} \rSL_{1/q'}^{\hspace{-0.6mm}(s)}\hspace{-0.4mm}(2,\C)$, 
mixing the chiral and antichiral sector, 
that one can naturally identify with the deformed real Lorentz group. 
This is the real form that corresponds to the real $q$-linked 
Minkowski space, that we will introduce in the next section.\\[2mm]

\begin{definition}
We define $\cU_{q,q'}^{\,\,\,\ell}(\mathrm{so}(3,1))$ to be the Hopf star algebra 
$\cU_{q,q'}(\rsl_2^{(s)} \stackrel{\ell}
\oplus \rsl_2^{(s)}), \stl)$
where
 $$
\begin{array}{rclcrcl}
(q^{\bh})^{\stl}&=&q'^{\bh'}&&(q'^{\bh'})^{\stl}&=&q^{\bh}\\[2mm]
(e_\pm)^{\stl}&=&-e'_{\pm}&&(e'_\pm)^{\stl}&=&-e_{\pm}\\[5mm]
&&&\bar{q}=\frac{1}{q'}&&&\\[2mm]
&&&\bar{\ell}=\ell&&&
\end{array}
$$
Similarly we define $\rSO_{q,q'}^{\,\,\,\ell}(3,1)$ as the Hopf star 
algebra  $(\rSL^{\hspace{-0.2mm}(s)}_q\hspace{-0.6mm}(2,\mathbb{C})\stackrel{\ell}{\oplus} \rSL^{\hspace{-0.2mm}(s)}_{1/q'}\hspace{-0.6mm}(2,\mathbb{C}), \srl)$ where 
$$
\begin{array}{rcl}
\begin{pmatrix} a&b\\c&d\end{pmatrix}^{\srl}&=&S\begin{pmatrix} a'&c'\\b'&d'\end{pmatrix}\\[5mm]
\bar{q}&=&\frac{1}{q'}\\[2mm]
\bar{\ell}&=&\ell
\end{array}
$$
\end{definition}

As expected, the star algebras defined above are in perfect duality;
we leave the easy verifications to the reader, they are coming
almost straighforwardly from Prop. \ref{dual}.
of our previous section.

\begin{proposition}
The dual paring given in (\ref{dual}) gives a perfect duality between the 
Hopf star algebras 
$\rSO_{q,q'}^{\,\,\,\ell}(3,1)$ 
and $\cU_{q,q'}^{\,\,\,\ell}(\mathrm{so}(3,1))$.

\end{proposition}

Analogously we can give a dual pairing between the various real
forms of the complex Lorentz group 
$\mathrm{SL}_q^{(s)}(2,\C)\stackrel{\ell}{\oplus}  \mathrm{SL}_{q}^{(s)}(2,\C)$ 
and the real forms of $\cU_{q,q'}(\rsl_2^{(s)} \stackrel{\ell}
\oplus \rsl_2^{(s)})$ defined as follows
(refer to Sec. \ref{realf-sec}):
\beq
\begin{array}{rl}
\cU_{q,q'}(\rsu(2)^{(s)} \stackrel{\ell}
\oplus \rsu(2)^{(s)})&:=\, (\cU_{q,q'}(\rsl_2^{(s)} \stackrel{\ell}
\oplus \rsl_2^{(s)}), \ste \oplus \ste) \\ \\
\cU_{q,q'}(\rsl_2(\R)^{(s)} \stackrel{\ell}
\oplus \rsl_2(\R)^{(s)})&:=\, (\cU_{q,q'}(\rsl_2^{(s)} \stackrel{\ell}
\oplus \rsl_2^{(s)}), \sts \oplus \sts) \\ \\
\cU_{q,q'}(\rsu(1,1)^{(s)} \stackrel{\ell}
\oplus \rsu(1,1)^{(s)})&:=\, (\cU_{q,q'}(\rsl_2^{(s)} \stackrel{\ell}
\oplus \rsl_2^{(s)}), \stk \oplus \stk) \\ \\
\cU_{q,q'}(\rsu(1,1)^{(s)} \stackrel{\ell}
\oplus \rsl_2(\R)^{(s)})&:=\, (\cU_{q,q'}(\rsl_2^{(s)} \stackrel{\ell}
\oplus \rsl_2^{(s)}), \stk \oplus \sts) \\ \\
\end{array}
\eeq

We summarize the results in the
following table and we give a remark on how to proceed for the
checks involved to prove the statements.

$$
\begin{tabular}{c|c|c|c}
QFA & QUEA & $\begin{array}{c} \hbox{star str.} \\
 \hbox{QFA} \\ \hbox{QUEA}\end{array}$
&  parameters\\
\hline& &\\
$\rSO_{q,q'}^{\,\,\,\ell}(3,1)$& $\cU_{q,q'}^{\,\,\,\ell}(\mathrm{so}(3,1))$&
$\begin{array}{c} \srl \\ \stl \end{array}$ 
&$\begin{array}{c} \bar{q}=q'^{-1}\,,\\
\bar{\ell}=\ell\end{array}$ \\ \hline&
\\
[-3.5mm]$\rSO_{q,q'}^{\,\,\,\ell}(4)$& 
$\cU_{q,q'}(\rsu(2)^{(s)} \stackrel{\ell}
\oplus \rsu(2)^{(s)})$
& 
$\begin{array}{c} \sre\oplus \sre\\ \ste \oplus \ste \end{array}$ 
&$\begin{array}{c}\bar{q}=q,\,\bar{q}'=q' \\
\bar{\ell}=\ell^{-1}\end{array}$ \\ \hline& &\\[-3.5mm]
$\rSO_{q,q'}^{\,\,\,\ell}(2,2)$
&
$\cU_{q,q'}(\rsl_2(\R)^{(s)} \stackrel{\ell}
\oplus \rsl_2(\R)^{(s)})$ &
$\begin{array}{c} \srs \oplus \srs\\ \sts \oplus \sts\end{array}$
&$\begin{array}{c}\bar{q}=q^{-1},\,\bar{q}'=\sfrac{1}{q'} ,\\
\bar{\ell}=\ell^{-1}\end{array}$ \\ \hline&&\\[-3.5mm]
$\rSO_{\,q,q'}^{'\,\,\,\ell}(2,2)$ &   
$\cU_{q,q'}(\rsu(1,1)^{(s)} \stackrel{\ell}
\oplus \rsu(1,1)^{(s)})$
&
$\begin{array}{c} \srk\oplus \srk\\ \stk \oplus \stk \end{array}$
&$\begin{array}{c}\bar{q}=q,\,\bar{q}'=q' ,\\
\bar{\ell}=\ell^{-1}\end{array}$ \\ \hline&&\\[-3.5mm]
$\rSO_{\,\,q,q'}^{''\,\,\,\ell}(2,2)$&   
$\cU_{q,q'}(\rsu(1,1)^{(s)} \stackrel{\ell}
\oplus \rsl_2(\R)^{(s)})$
&
$\begin{array}{c} \srk\oplus \srs\\ \stk \oplus \sts \end{array}$
&$\begin{array}{c}\bar{q}=q,\,\bar{q}'=\sfrac{1}{q'} ,\\
\bar{\ell}=\ell^{-1}\end{array}$ \\ \hline
\end{tabular}
$$
Note that this table reproduces the central column of table 1 of \cite{luk}. In our notation $*_a\oplus*_b$ and $\star_a\oplus\star_b$ with $a,b=e,s,k$ are the involutions  acting with $*_a$ or $\star_a$ on the chiral sector and with $*_b$ or $\star_b$ on the antichiral one.

We now make some remarks on how the reader can go about the
proof of the dualities expressed in the previous table.

\begin{remark}
When we look at the pairing $\langle x, f\rangle$ with both $x$ and $f$ in the chiral or antichiral sector, the compatibility of the star structures with the dual pairing is guaranteed by the results presented in the previous section. To the only checks involve expressions of type $\langle q^{\bh},a'\rangle$. Note that the Euclidean Klenian and split involutions map $q^{\bh}$ to itself and $a'$ to itself or $d'$, and since $\langle q^{\bh},a'\rangle=1=\langle q^{\bh},d'\rangle$ the consistency of those star structures with the pairing proposed (that is $
\langle x^*,f\rangle=\overline{\langle x,S(f)^\star\rangle}, $  ) follows naturally. 

The compatibility of the Lorentzian star structure can be then proved case by case; as sample calculation we check one mixed term. according to our definition we have
$$
\langle (q^{\bh})^{*_{\textrm{\tiny{l}}}}, a'\rangle=\langle q'^{\bh'}, a'\rangle=(q')^{-\frac 1 2}
$$ 
while
$$
\langle q^{\bh}, S(a')^{{\srl}}\rangle=\langle q^{\bh}, d'^{{\srl}}\rangle=\langle q^{\bh}, a\rangle=q^{\frac 1 2}
$$
implying that $\bar{q}=\frac{1}{q'}$ as required by hypothesis. 
\end{remark}




\section{The $q$-linked Minkowski, Klein and Euclidean spaces}

In this section we want to define the complex quantum Minkowski space as
homogeneous space for the complex Lorentz quantum group, introduced in the previous
section. Then, we will examine its real forms: the real Minkowski, the Klein and
Euclidean quantum spaces, together with their isometry groups.


\subsection{The complex $q$-linked Minkowski space}

In \cite{fl} Ch. 4, a one parameter quantum deformation of the complex Minkowski
space is introduced by giving a deformation of the Poincar\'e group and identifying
the Minkowski space with its subgroup of translations, so to obtain a natural
a quantum homogeneous space structure on it. 
We now want to generalize this construction to a multiparameter
quantum deformation, however focusing on the complex Lorentz
group coaction only.

\begin{definition}
We define the \textit{complex $q$-linked Minkowski space} 
$\mathcal{M}^{\ell,q}_{\mathbb{C}}$ as the 
complex algebra generated by the indeterminates $t_{ij}$, $i,j=1,2$ subject to the relations:
$$
\begin{array}{rclccrcl}
t_{11}t_{12}&=&q\ell \,t_{12}t_{11}&& t_{12}t_{22}&=&q^{-1}\ell \,t_{22}t_{12}\\[3mm]
t_{12}t_{21}&=&t_{21}t_{12}+\ell (q^{-1}-q)t_{11}t_{22}&&t_{11}t_{22}&=&t_{22}t_{11}\\[3mm]
t_{11}t_{21}&=&q^{-1}\ell^{-1}\,t_{21}t_{11}&&t_{21}t_{22}&=&q\ell^{-1}\,t_{22}t_{21}
\end{array}
$$
with a quantum metric (see \cite{fi2}):
$$
Q_{\ell/q}=t_{12}t_{21}-(q^{-1}\ell)\, t_{11}t_{22}
$$
(For the meaning of $q$ and $\ell$ refer to the previous sections).
\end{definition}

The following proposition is a straightforward tedious calculation.

\label{coaction}\begin{proposition}
The complex linked Minkowski space admits a natural coaction of the complex $q$-linked
Lorentz  
group $\rSL_q^{\hspace{-0.6mm}(s)}\hspace{-0.6mm}(2,\C)\stackrel{\ell}{\oplus} \rSL_{q}^{\hspace{-0.6mm}(s)}\hspace{-0.6mm}(2,\C)$
$$
\begin{array}{ccccccl}
\rho&:&\mathcal{M}_{\C}^{\ell,q} & \longrightarrow & \rSL_q^{\hspace{-0.6mm}(s)}\hspace{-0.6mm}(2,\C)\stackrel{\ell}{\oplus} \rSL_{q}^{\hspace{-0.6mm}(s)}\hspace{-0.6mm}(2,\C) &\otimes& \mathcal{M}_{\C}^{\ell,q} \\[3mm] 
&&t_{ij} & \longmapsto &\sum_{s,r} y_{is}S(x_{rj}) &\otimes &t_{sr}
\end{array}
$$
Moreover such coaction preserves the quantum metric $Q_{\ell/q}$, that is
$$
\rho:Q_{\ell/q}\longrightarrow 1\otimes Q_{\ell/q}
$$
\end{proposition}

We now want to give some physical interpretation of the
deformation $\mathcal{M}^{\ell,q}_{\mathbb{C}}$.
Consider the two quantum planes
$\C_q[\chi_1,\chi_2]$ and $\C_{q}[\psi_1,\psi_2]$.
These are the algebras generated over the ring $\C_q=\C[q,q^{-1}]$
($q$ here interpreted as an indeterminate)
by the elements $\chi_i$, $\psi_i$, $i=1,2$ subject to the relations:
\beq\label{qplanes}
\chi_1\chi_2=q^{-1}\chi_2\chi_1, \qquad \psi_1\psi_2=q^{-1}\psi_2\psi_1
\eeq
We also refer to a quantum plane with the notation $\C_q^2$.

\begin{definition}
We define \textit{$q$-linked planes} $\C_q^2\stackrel{\ell}{\oplus}
\C_q^2$ as the $\C_q$ algebra generated by 
$\chi_i$, $\psi_i$, $i=1,2$ subject to the relations (\ref{qplanes})
and the additional \textit{linked relations}:
\beq\label{linkedrel}
\begin{array}{rclccrcl}
\chi_1\psi_2&=&\ell^{-\frac 1 2}\psi_2\chi_1&&\chi_2\psi_2&=&\ell^{\frac 1 2}\psi_2\chi_2\\[3mm]
\chi_1\psi_1&=&\ell^{\frac 1 2}\psi_1\chi_1&&\chi_2\psi_1&=&\ell^{-\frac 1 2}\psi_1\chi_2\\[3mm]
\end{array}
\eeq
\end{definition}

A quantum plane $\C_q^2$ is naturally both
a left comodule and a right comodule
(via the antipode) for $\rSL_q(2,\C)$ (see \cite{ma1, fl}).
It is only natural to expect that such comodule structures
are inherited by the $q$-linked planes, so that we have a
coaction of the complex Lorentz quantum group.

In the next proposition, we establish a relation between the comodule structure
of the $q$-linked planes and the comodule structure of the
quantum Minkowski space, according to the separate actions of  
the chiral and antichiral sectors discussed previously.

\begin{proposition}
There is a $\rSL_q^{\hspace{-0.6mm}(s)}\hspace{-0.6mm}(2,\C)\stackrel{\ell}{\oplus} \rSL_{q}^{\hspace{-0.6mm}(s)}\hspace{-0.6mm}(2,\C)$
comodule structure on $\C_q^2\stackrel{\ell}{\oplus}
\C_q^2$ given by:
$$
\begin{array}{ccc}
\C_q^2\stackrel{\ell}{\oplus}
\C_q^2 & \lra & 
\rSL_q^{\hspace{-0.6mm}(s)}\hspace{-0.6mm}(2,\C)\stackrel{\ell}{\oplus} 
\rSL_{q}^{\hspace{-0.6mm}(s)}\hspace{-0.6mm}(2,\C) \,\,
\otimes \, \, \C_q^2\stackrel{\ell}{\oplus} \C_q^2 \\ \\
\begin{pmatrix} \chi_1\\ \chi_2 \end{pmatrix} & \mapsto & 
\begin{pmatrix} y_{11} & y_{12} \\ y_{21} & y_{22} 
\end{pmatrix} \otimes \begin{pmatrix} \chi_1\\ \chi_2 \end{pmatrix}, \\ \\
\begin{pmatrix} \psi_2\\ -q\psi_1 \end{pmatrix} & 
\mapsto & S\begin{pmatrix} x_{11} & x_{12} \\
x_{21} & x_{22} \end{pmatrix} \otimes \begin{pmatrix} 
\psi_2\\ -q\psi_1 \end{pmatrix}, \\ \\
\end{array}
$$
Furthermore, we have a $\rSL_q^{\hspace{-0.6mm}(s)}\hspace{-0.6mm}(2,\C)\stackrel{\ell}{\oplus} \rSL_{q}^{\hspace{-0.6mm}(s)}\hspace{-0.6mm}(2,\C)$
comodule morphism:
$$
\begin{array}{ccc}
\C_q^2\stackrel{\ell}{\oplus}
\C_q^2 & \lra & \mathcal{M}^{\ell,q}_{\mathbb{C}} \\ \\
\begin{pmatrix} \chi_1\\\chi_2 \end{pmatrix}\otimes  \begin{pmatrix} 
\psi_2 & -q\psi_1 \end{pmatrix}
& \mapsto &\begin{pmatrix} t_{11} & t_{12} \\
t_{21} & t_{22}
\end{pmatrix} 
\end{array}
$$
\end{proposition}

\begin{proof} Direct calculation. \end{proof}

Notice that, with this interpretation, we view the complex Minkowski
space as a tensor product of two spinor representations: one chiral
the other one antichiral.


\subsection{Real forms of the $q$-linked complex Minkowski space}

We now turn to examine the real forms of $\mathcal{M}^{\ell,q}_{\mathbb{C}}$, that
is the star structures we can impose on $\mathcal{M}^{\ell,q}_{\mathbb{C}}$. 
For each real form we provide its isometry quantum group, that is, we
provide a real form of the $q$-linked complex Lorentz group compatible with
the complex coaction $\rho$.

The proofs
of all the propositions appearing in this section amount just to  tedious checkings, so
we leave them to the reader.

\bigskip
{\bf The real $q$-linked Minkowski space}. We define
the {\it real $q$-linked Minkowski space}
 $\mathcal{M}^{\ell,q}_{3,1}$ as the pair $(\mathcal{M}^{\ell,q}_{\mathbb{C}},\medstar_{\textrm{m}})$ with $\medstar_{\textrm{m}}$ being the star structure: 
\beq\label{starmink2}
\begin{array}{cccc}
\medstar_{\textrm{m}}: & \mathcal{M}^{\ell,q}_{\mathbb{C}} & \longrightarrow & \mathcal{M}^{\ell,q}_{\mathbb{C}} \\
&\begin{pmatrix}t_{11} & t_{12} \\ t_{21} & t_{22} \end{pmatrix} &
\longmapsto & \begin{pmatrix}t_{11} & t_{21} \\ t_{12} & t_{22} \end{pmatrix}\\[4mm]
& q & \longmapsto & q\\
& \ell & \longmapsto & \ell
\end{array}
\eeq
\medskip
\begin{proposition}
The real linked Minkowski space is a quantum homogeneous 
space with respect to the coaction of the quantum group $\rSO_{q}^{\ell}(3,1)$. In other words we have:
$$
\rho \circ \medstar_{\textrm{m}}=(\srl \otimes \medstar_{\textrm{m}})\circ \rho
$$

\end{proposition}

\begin{remark}
Notice that in order to get a well defined coaction on the $q$-linked real Minkowski space,
 we are forced to use a real form of $\rSL^{\hspace{-0.2mm}(s)}_q\hspace{-0.6mm}(2,\mathbb{C})\stackrel{\ell}{\oplus} \rSL^{\hspace{-0.2mm}(s)}_{q}\hspace{-0.6mm}(2,\mathbb{C})$. In
other words we are forced to specialize the parameter $q'$ to $q$.
\end{remark} 

\bigskip
{\bf The $q$-linked Euclidean space}.
In analogy with \cite{flm}, let
$\rSO_{q}^{\ell}(4)$ be the Hopf star algebra  
$\big(\rSL^{\hspace{-0.2mm}(s)}_q\hspace{-0.6mm}(2,\mathbb{C})\stackrel{\ell}{\oplus} \rSL^{\hspace{-0.2mm}(s)}_{q}\hspace{-0.6mm}(2,\mathbb{C}),\sre\oplus \sre\big)$.
We define the {\it real $q$-linked Euclidean space} $\mathcal{M}^{\ell,q}_{4}$ 
to be the pair $(\mathcal{M}^{\ell,q}_{\mathbb{C}},\medstar_{\textrm{e}})$ with $\medstar_{\textrm{e}}$ the star structure given by:
\beq\label{stareucl}
\begin{array}{cccc}
\medstar_{\textrm{e}}: & \mathcal{M}^{\ell,q}_{\mathbb{C}} & \longrightarrow & \mathcal{M}^{\ell,q}_{\mathbb{C}} \\
&\begin{pmatrix}t_{11} & t_{12} \\ t_{21} & t_{22} \end{pmatrix} &
\longmapsto & \begin{pmatrix}-q^{-1}\ell t_{22} & t_{21} \\ t_{12} & -q\ell t_{11} \end{pmatrix}\\[4mm]
& q & \longmapsto & q\\
& \ell & \longmapsto & \ell^{-1}
\end{array}
\eeq

\begin{proposition}
The real linked Euclidean space is a quantum homogeneous with respect to the coaction of the quantum group $\rSO_{q}^{\ell}(4)$. In other words we have:
$$
\rho \circ \medstar_{\textrm{e}}=((\sre\oplus \sre) \otimes \medstar_{\textrm{e}})\circ \rho
$$
\end{proposition}

\bigskip
{\bf The $q$-linked Klein space}.
As already noticed by \cite{luk}, we have two 
real forms of the quantum group  $\rSL_q^{\hspace{-0.6mm}(s)}\hspace{-0.6mm}(2,\C)\stackrel{\ell}{\oplus} \rSL_{q}^{\hspace{-0.6mm}(s)}\hspace{-0.6mm}(2,\C)$,
both deformations of the Klein isometry group, namely:
$\rSO_{q}^{\ell}(2,2):=\rSL^{\hspace{-0.2mm}(s)}_q\hspace{-0.6mm}(2,\mathbb{R})
\oplus \rSL^{\hspace{-0.2mm}(s)}_q\hspace{-0.6mm}(2,\mathbb{R})$ and $
\rSO_{q}'^{\ell}(2,2):=\rSU^{\hspace{-0.2mm}(s)}_q\hspace{-0.6mm}(1,1)
\oplus \rSU^{\hspace{-0.2mm}(s)}_q\hspace{-0.6mm}(1,1)$.

Notice that, in principle, we could have a third possibility, mixing the
two components, however in this case, it is not possible to construct a meaningful
coaction on a suitable space, so we discard it, being outside our scope.

\medskip
Let us define the {\it real $q$-linked Klein space} $\mathcal{M}^{\ell,q}_{2,2}$ as  the pair $(\mathcal{M}^{\ell,q}_{\mathbb{C}},\medstar_{\textrm{s}})$ with $\medstar_{\textrm{s}}$ being the star structure given by:
\beq\label{starklein1}
\begin{array}{cccc}
\medstar_{\textrm{s}}: & \mathcal{M}^{\ell,q}_{\mathbb{C}} & \longrightarrow & \mathcal{M}^{\ell,q}_{\mathbb{C}} \\
&\begin{pmatrix}t_{11} & t_{12} \\ t_{21} & t_{22} \end{pmatrix} &
\longmapsto & \begin{pmatrix}t_{11} & t_{12} \\ t_{21} & t_{22} \end{pmatrix} \\[4mm]
& q & \longmapsto & q\\
& \ell & \longmapsto & \ell
\end{array}
\eeq
Let us remember that when dealing with the real form $\rSO_{q}^{\ell}(2,2)$ it is convenient to work with a multiplicative involution and real parameters $q$ and $\ell$ and this is consitent with the previous construction.\\[2mm]

We also define the {\it real $q$-linked Klein space} $\mathcal{M}^{'\,\ell,q}_{\,\,2,2}$  as  the pair $(\mathcal{M}^{\ell,q}_{\mathbb{C}},\medstar_{\textrm{k}})$ with $\medstar_{\textrm{k}}$ being the star structure:
\beq\label{starklein2}
\begin{array}{cccc}
\medstar_{\textrm{k}}: & \mathcal{M}^{\ell,q}_{\mathbb{C}} & \longrightarrow & \mathcal{M}^{\ell,q}_{\mathbb{C}} \\
&\begin{pmatrix}t_{11} & t_{12} \\ t_{21} & t_{22} \end{pmatrix} &
\longmapsto & \begin{pmatrix}q^{-1}\ell t_{22} & t_{21} \\ t_{12} & q\ell t_{11} \end{pmatrix}\\[4mm]
& q & \longmapsto & q\\
& \ell & \longmapsto & \ell^{-1}
\end{array}
\eeq


\begin{proposition}
The real $q$-linked Klein spaces $\mathcal{M}^{'\,\ell,q}_{\,\,2,2}$  and $\mathcal{M}^{\ell,q}_{2,2}$  are quantum homogeneous spaces with respect to the coaction of the quantum groups $\rSO_{q}^{\ell}(2,2)$ and $\rSO_{\,\,q}^{'\,\ell}(2,2)$ that is:
$$
\rho \circ \medstar_{\textrm{s}}=((\srs\oplus \srs) \otimes \medstar_{\textrm{s}})\circ \rho
$$
and
$$
\rho \circ \medstar_{\textrm{k}}=((\srk\oplus \srk) \otimes \medstar_{\textrm{k}})\circ \rho
$$
\end{proposition}

We conclude this section with some remarks and comments.

\begin{remark}
\begin{enumerate}
\item 
One may try to construct a three parameter deformation ($q$ $q'$ and $\ell$)
of the complex Minkowski space, however, then trying to construct a
coaction of the $q$-linked complex 4-d orthogonal group
$\rSL_q^{\hspace{-0.6mm}(s)}\hspace{-0.6mm}(2,\C)\stackrel{\ell}{\oplus} \rSL_{1/q'}^{\hspace{-0.6mm}(s)}\hspace{-0.4mm}(2,\C)$, the parameter $q'$ will be forced
to specialize to $q$.  

A similar, but different, phenomenon occurs when trying to construct 
the real quantum Klein space.
as quantum homogeneous space with respect the coaction of 
$\rSO_{\,\,q}^{''\,\ell}(2,2)$. In this case we cannot specialize $q'$ to $q'$ 
due to the star structure involved, hence we are unable to provide
such deformation. It is not clear at the moment the possible physical interpretation of such constraints and we plan to tackle this interesting question
in a future project.

\item As one can readily check, 
all the star structures on $\mathcal{M}^{\ell,q}_{\mathbb{C}}$ preserve
the metric $Q_{\ell/q}$ and consequently the coaction of the
symmetry group of each of our real form will also preserve such
metric.
\end{enumerate}

\end{remark}

\section{Conclusions}
We have provided QUEA and QFA deformations of the complex orthogonal
group in four dimensions together with its real forms, thus obtaining an
enlarged version of the Table 1 in \cite{luk}.
For each of the real forms (excluding
the quaternionic one), we have realized a quantization of the corresponding
real homogeneous space: Minkowski, Euclidean and Klein, where the real
form of the complex orthogonal group appears as the isometry group, i.e.
preserves the metric.

Hence, according to the results
in [12], we are classifying the real forms of the complex orthogonal
group in four dimensions, which admit a coaction on real forms
of the complex Minkowski space in the standard-standard case only
(see Sec. 5.1).

We are however forced to reduce our three parameter deformation to a
two parameter one. We plan to examine these missing cases in a future
paper.
We summarize our constructions in the following table:\\[5mm]
\bigskip
\begin{center}
\begin{tabular}{c|c|c|c}
Isometry group & Q. Homo. Space & 
\scalebox{1.9}{$\star$}& Eq. no.\\\hline&&\\[-3.5mm]
$\rSO_{q}^{\,\,\,\ell}(3,1)$& $\mathcal{M}^{\ell,q}_{3,1}$
&$\medstar_m$& (\ref{starmink2}) 
\\ 
\hline
$\rSO_{q}^{\,\,\,\ell}(4)$& $\mathcal{M}^{\ell,q}_4$
&$\medstar_e$& (\ref{stareucl})\\
\hline
$\rSO_{q}^{\,\,\,\ell}(2,2)$& $\mathcal{M}^{\ell,q}_{2,2}$
&$\medstar_s$& (\ref{starklein1})\\
\hline
$\rSO_{q}'^{\,\,\,\ell}(2,2)$& $\mathcal{M'}^{\ell,q}_{2,2}$
&$\medstar_k$& (\ref{starklein2})\\
\hline
\end{tabular}
\end{center}

\end{document}